\documentclass[10pt]{article}

\usepackage[T1]{fontenc}
\usepackage[utf8]{inputenc}

\usepackage[english]{babel}

\usepackage{amssymb,amsthm,mathtools}
\usepackage{amsfonts,dsfont}

\usepackage{lmodern}

\usepackage[text={122mm,193mm},centering]{geometry}
\usepackage{csquotes,microtype,xcolor}
\usepackage{graphicx,siunitx,setspace}

\usepackage{wrapfig,caption}
\usepackage[numbers]{natbib}
\usepackage[unicode,colorlinks,bookmarksnumbered]{hyperref}
\hypersetup{linkcolor=blue,citecolor=blue,urlcolor=blue}

\usepackage{orcidlink}

\usepackage[capitalise]{cleveref}
\crefname{figure}{Fig.}{Figs.}

\usepackage{tikz,algorithm,algpseudocode}

\theoremstyle{plain}
\newtheorem{theorem}{Theorem}

\theoremstyle{definition}

\theoremstyle{remark}

\renewcommand*{\le}{\ensuremath \leqslant} 
\renewcommand*{\ge}{\ensuremath \geqslant}

\newcommand*{\R}{\ensuremath \mathbb{R}}
\newcommand*{\dr}{\ensuremath \partial}
\newcommand*{\intd}[1]{\ensuremath \mathrm{d}#1}
\newcommand*{\segn}{\ensuremath {\{1,\dots,n\}}}

\newcommand*{\koni}{\ensuremath k_{\text{on},i}}
\newcommand*{\koffi}{\ensuremath k_{\text{off},i}}
\newcommand*{\Exp}{\ensuremath \mathcal{E}} 
\newcommand*{\Prob}{\ensuremath \mathbb{P}} 
\newcommand*{\Esp}{\ensuremath \mathbb{E}} 
\newcommand*{\Yold}{\ensuremath Y_{\text{old}}} 
\newcommand*{\Zold}{\ensuremath Z_{\text{old}}} 
\newcommand*{\Told}{\ensuremath T_{\text{old}}} 
\newcommand*{\flotY}{\ensuremath \varphi_{\text{M}}} 
\newcommand*{\flotZ}{\ensuremath \varphi_{\text{P}}} 
\newcommand{\AltTextCMSB}[1]{} 

\title{\vspace{-18mm}\textbf{Efficient stochastic simulation of gene regulatory networks using hybrid models of transcriptional bursting}}

\author{\normalsize Mathilde Gaillard\,\orcidlink{0009-0002-9520-3814} and Ulysse Herbach\,\orcidlink{0000-0002-0972-385X}}

\date{\normalsize Université de Lorraine, CNRS, Inria, IECL, F-54000 Nancy, France\\\small\texttt{\{mathilde.gaillard, ulysse.herbach\}@inria.fr}}

\begin{document}

\maketitle

\begin{center}
\begin{minipage}{10.25cm}
\small\noindent\textbf{Abstract.}
Single-cell data reveal the presence of biological stochasticity between cells of identical genome and environment, in particular highlighting the transcriptional bursting phenomenon. To account for this property, gene expression may be modeled as a continuous-time Markov chain where biochemical species are described in a discrete way, leading to Gillespie's stochastic simulation algorithm (SSA) which turns out to be computationally expensive for realistic mRNA and protein copy numbers. Alternatively, hybrid models based on piecewise-deterministic Markov processes (PDMPs) offer an effective compromise for capturing cell-to-cell variability, but their simulation remains limited to specialized mathematical communities. With a view to making them more accessible, we present here a simple simulation method that is reminiscent of SSA, while allowing for much lower computational cost. We detail the algorithm for a bursty PDMP describing an arbitrary number of interacting genes, and prove that it simulates exact trajectories of the model. As an illustration, we use the algorithm to simulate a two-gene toggle switch: this example highlights the fact that bimodal distributions as observed in real data are not explained by transcriptional bursting per se, but rather by distinct burst frequencies that may emerge from interactions between genes.

\vspace{3.2mm}

\end{minipage}\medskip
\end{center}

\section{Introduction}

A common simplification of the gene expression mechanism considers that a gene is transcribed into mRNA, which in turn is translated into proteins. With the advent of technologies that can measure this process at the level of individual cells and individual molecules, it has become clear that the first step is not as smooth as previously expected from bulk data~\cite{Chubb2010}. In particular, mRNA synthesis typically occurs in bursts, i.e., short but intense periods of transcription during which many mRNA copies are produced~\cite{Raj2006,Rodriguez2020} (corresponding to convoys of closely spaced polymerases~\cite{Tantale2016,Wang2019}), interspersed with inactive periods: these transcriptional bursts generate intrinsic stochasticity~\cite{Raj2006,Nicolas2017} with functional implications for the cell~\cite{Raj2008}.

In the context of gene regulatory networks, this fundamental property should therefore not be ignored, especially when considering single-cell data~\cite{Herbach2017,Ventre2023}. It is worth noticing that single-cell RNA-seq data, which measures the expression of many genes simultaneously but is usually seen as very noisy, turns out to be largely consistent with much more precise techniques such as RT-ddPCR~\cite{Albayrak2016} and smFISH~\cite{Raj2006,Singer2014}. In previous work, we introduced a mechanistic approach to gene regulatory network inference~\cite{Herbach2017} that was later implemented as a Python package~\cite{Herbach2023}. This approach can be interpreted as calibration of a particular dynamical model describing an arbitrary number of interacting genes, formulated as a piecewise-deterministic Markov process (PDMP) driven by transcriptional bursting: once calibrated, the model can be used to simulate single-cell data and is able to reproduce the biological variability observed experimentally~\cite{Ventre2023}.

Here we focus on exact simulation of this model. Indeed, while PDMPs are being used increasingly as a powerful hybrid framework for modeling various biological phenomena~\cite{Rudnicki2017}, their simulation is not yet as widespread as their fully discrete counterpart, the well-known Doob–Gillespie algorithm, also known as stochastic simulation algorithm (SSA, see e.g.~\cite{Schnoerr2017} for an introduction). The aim of this article is to present a simple algorithm capable of generating exact trajectories: it is reminiscent of SSA but introduces an acceptance-rejection step which, although sub-optimal in general, is straightforward to implement and considerably reduces the computational cost. As an underlying message, we argue that transcriptional bursting represents an ideal situation for this type of hybrid model, where abundant chemical species are described using ordinary differential equations (ODEs) corresponding to mass action kinetics, but affected by discrete events related to rare species.

The algorithm we propose is not new per se: it can be seen as a particular case of a more general framework for simulating PDMPs~\cite{Lemaire2018}. Nevertheless, the present paper is motivated by three main reasons.

First, we provide a particular instance that is dedicated to modeling gene regulatory networks, aiming at a wider adoption of this framework to simulate large networks efficiently, as an alternative to exact but inefficient SSA, or faster but approximate SSA variants. We emphasise that the bursty PDMP is the simplest model capable of correctly capturing the specific variability (namely, distribution tails) observed in current single-cell data.

Second, we adopt a simplified viewpoint compared to more theoretically oriented papers~\cite{Benaim2015,Lemaire2018}, in order to highlight the feasibility and modularity of this framework in the context of systems biology. In particular, the proposed procedure can be easily extended beyond the two-stage mRNA-protein model, for example by distinguishing mature mRNA and pre-mRNA~\cite{Rudnicki2015}.

Finally, we could not find a detailed proof of the algorithm in the literature, that is, a proof that the construction procedure simulates exact trajectories of the model. In more mathematical terms, one needs to show that the constructed Markov process actually has the correct infinitesimal generator. 
The procedure is similar to the more common continuous-time Markov chains but needs a specific treatment, due to the fact that the state of the system between jumps is not constant. We therefore provide a proof here for completeness.

This paper is organized as follows. In \cref{sec_model} we introduce the model, and then detail the simulation algorithm in \cref{sec_algorithm}. The interest of this model is discussed in \cref{sec_discussion}, and a proof of the algorithm is given in \cref{sec_proof}.

\section{Mathematical model}
\label{sec_model}

In the following, we consider a network of $n$ genes possibly in interaction. We start with the well-known ‘two-state model’ of gene expression, which corresponds to the following set of elementary chemical reactions for each gene~$i \in \segn$:
\begin{equation}
\label{eq_reactions}
\begin{array}{ccc}
G_i \xrightleftharpoons[\koffi]{\koni} {G_i}^*,\,\quad & {G_i}^* \xrightarrow[]{s_{0,i}} {G_i}^* + M_i,\,\quad & M_i \xrightarrow[]{d_{0,i}} \varnothing , \\
& M_i \xrightarrow[]{s_{1,i}} M_i + P_i,\,\quad & P_i \xrightarrow[]{d_{1,i}} \varnothing ,
\end{array}
\end{equation}
where $G_i$, ${G_i}^*$, $M_i$ and $P_i$ denote inactive promoter (off), active promoter (on), mRNA and protein for gene~$i$.
These reactions describe the main two stages of gene expression, namely \emph{transcription} (rate $s_{0,i}$) and \emph{translation} (rate $s_{1,i}$), along with degradation of mRNA (rate $d_{0,i}$) and protein (rate $d_{1,i}$) molecules. Note that $G_i$ and ${G_i}^*$ satisfy a conservation relation and only one molecule of $G_i$ or ${G_i}^*$ is present at a time: for convenience we define $E_i \in \{0,1\}$ based on ${G_i}^*$.
The particularity of the two-state model is that transcription can only occur when the promoter is active, i.e., $E_i=1$ (\cref{fig1}A). The bursty regime corresponds to $\koffi \gg \koni$, meaning that promoter active periods are short compared to inactive periods (with $s_{0,i} \gg d_{0,i}$ such that $\koni s_{0,i}/(\koffi d_{0,i})$ is fixed). In the limiting case, active periods become infinitely short so the promoter state is not needed anymore and the gene activity is characterized by its \emph{burst frequency} $\koni$ (\cref{fig1}B).

Although regulation can take different forms~\cite{Molina2013,Sanchez2013}, modulation of burst frequency appears as a fundamental mechanism for possible interactions between genes~\cite{Larson2013,Nicolas2018,Wang2019,Jeziorska2022}. We therefore assume that $\koffi$ is constant, while $\koni$ is a function of protein levels $P = (P_1, \dots, P_n)$. In our paradigm, the bursts themselves are thus not coordinated directly (which would be relevant for genes in the same genomic location~\cite{Raj2006}), but interactions arise through functions $\koni$ which introduce correlation of burst frequencies.
For example, in~\cite{Herbach2023,Ventre2023} these functions take the following form:
\[\koni(P) = \frac{k_{0,i} + k_{1,i}\exp(\beta_i + \sum_{j=1}^n\theta_{ji}P_j)}{1 + \exp(\beta_i + \sum_{j=1}^n\theta_{ji}P_j)} \qquad\forall i\in\segn\]
with $k_{0,i} \ll k_{1,i}$ so that $\beta_i$ encodes basal activity of gene $i$ and $(\theta_{ij})_{1\le i,j \le n}$ can be interpreted as a network interaction matrix.

\begin{figure}[ht]
\includegraphics[width=\textwidth]{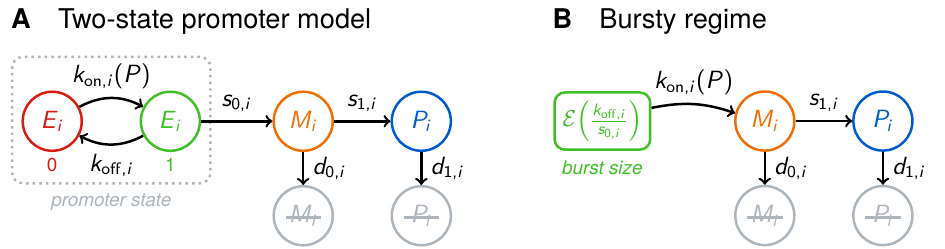}
\caption{General mechanistic model for the expression of a gene $i$, used as a basic unit for gene regulatory networks made of $n$ genes. Variables $M_i$ and $P_i$ denote respectively mRNA and protein quantities associated to gene $i$ in the cell. The common two-state promoter model (A) can be simplified in the bursty regime ($\koni \ll \koffi$) by discarding the promoter state $E_i$. The resulting model (B) describes instantaneous bursts of mRNA production, occurring at random times with rate $\koni(P)$ where $P=(P_1,\cdots,P_n)$ is the vector of protein levels. In this paradigm, the genes are indirectly coupled together through interaction functions $\koni$ which modulate burst frequencies.}
\label{fig1}
\AltTextCMSB{General mechanistic model for stochastic gene expression, used as a basic unit for gene regulatory networks.}
\end{figure}

When modeled as a PDMP (see~\cite{Herbach2017} and references therein), the dynamics are given by
\begin{equation}\label{eq_model_twostate}
\left\{\begin{array}{rl}
E_i(t) &\hspace{-2mm} : 0 \xrightarrow{\koni} 1,\,\; 1 \xrightarrow{\koffi} 0 \vspace{1.5mm} \\
{M_i}'(t) &\hspace{-2mm} = s_{0,i}E_i(t) - d_{0,i} M_i(t) \vspace{1.5mm} \\
{P_i}'(t) &\hspace{-2mm} = s_{1,i}M_i(t) - d_{1,i}P_i(t)
\end{array}\right.
\end{equation}
for each gene $i \in \segn$, where we omit the dependence of $\koni$ on $P$ for clarity. An example of stochastic trajectory is given in \cref{fig2} (bottom left), to be compared with usual SSA (\cref{fig2}, top left).
This is the model being considered in~\cite{Raj2006,Herbach2017,Bonnaffoux2019,Koshkin2024}.
It is also similar to that of~\cite{Kurasov2018}, with two differences: (1) we describe explicitly mRNA levels (not only protein levels) and (2) we consider that interactions arise only from $\koni$ and most importantly are not linear (i.e., $\koni$ is a nonlinear, bounded function of $P$): this is an important aspect that is necessary to match real data (see \cref{sec_discussion}).

\begin{figure}[ht]
\includegraphics[width=\textwidth]{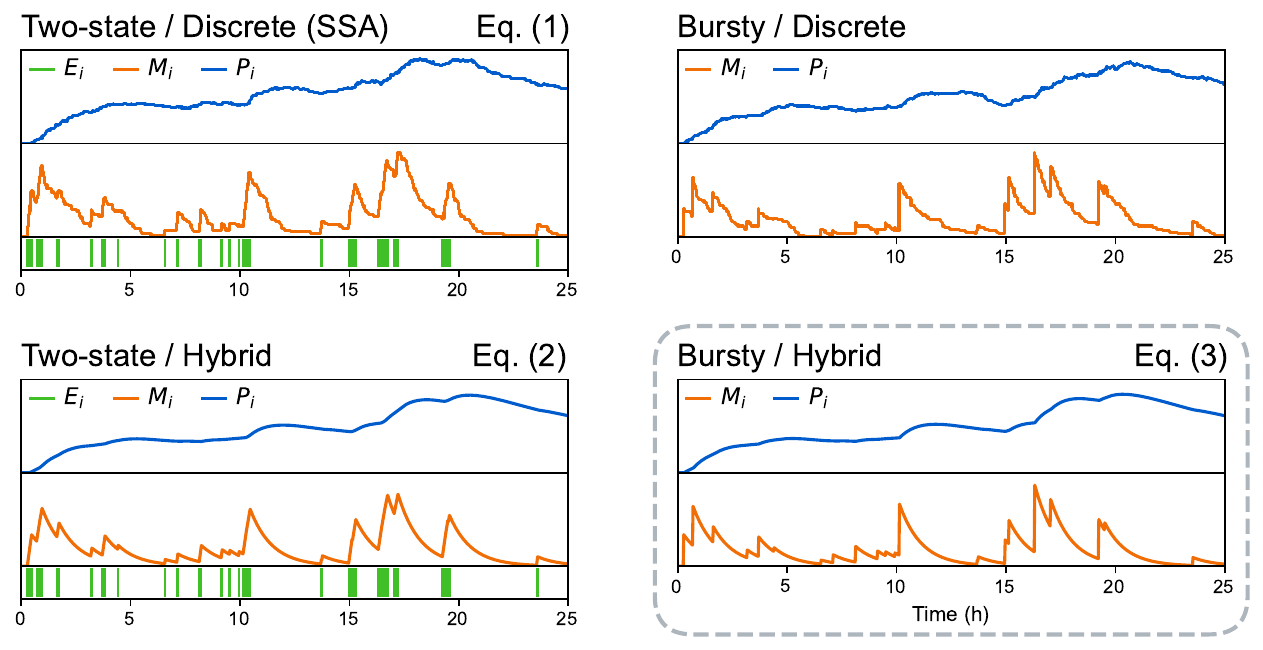}
\caption{Different mathematical frameworks for the stochastic gene expression model described in \cref{fig1}. The basic SSA (top left) corresponds to a continuous-time Markov chain implementing the stochastic mass action kinetics~\eqref{eq_reactions}. The bursty version (top right) removes the promoter description and is driven by instantaneous bursts following a geometric distribution, while the hybrid version (bottom left) describes continuous levels of mRNA and protein, forming a PDMP. The model we consider here is another PDMP that combines these two simplifications (bottom right). The burst sizes then follow an exponential distribution as shown in \cref{fig1}B.
To allow for comparison between models, the bursts are synchronized and the ranges of mRNA and protein in the four trajectories are identical (average $\langle M_i \rangle \approx 12$ copies/cell and $\langle P_i \rangle \approx 78$ copies/cell). Note that protein copy numbers are often much larger ($ P_i \sim 10^4$~\cite{Schwanhausser2011} or $10^5$~\cite{Albayrak2016}) so that discrete trajectories become indistinguishable from their continuous counterparts.}
\label{fig2}
\AltTextCMSB{Different mathematical frameworks for the stochastic gene expression model.}
\end{figure}

In the bursty regime, the promoter state disappears, mRNA is produced by instantaneous bursts and proteins are produced accordingly (see~\cite{Mackey2013,Chen2019} for mathematical details on the relation between the discrete and continuous models). In the PDMP setting, this results in the following model (still omitting the dependence of $\koni$ on $P$):
\begin{equation}\label{eq_model_bursty}
\left\{\begin{array}{rl}
M_i(t) &\hspace{-2mm} \xrightarrow{\koni} M_i(t) + \Exp\Big(\frac{\koffi}{s_{0,i}}\Big) \vspace{0.5mm} \\
{M_i}'(t) &\hspace{-2mm} = -d_{0,i} M_i(t) \vspace{1.5mm} \\
{P_i}'(t) &\hspace{-2mm} = s_{1,i} M_i(t) - d_{1,i} P_i(t)
\end{array}\right.
\end{equation}
for which an example of stochastic trajectory is given in \cref{fig2} (bottom right), to be compared with bursty SSA (\cref{fig2}, top right). This model is used in~\cite{Ventre2023,Nguyen2024}. A similar model is considered in~\cite{Pajaro2017} (using a PDE approach) with the difference that here we also explicitly model mRNA in order to fit single-cell transcriptomic data (see also \cref{sec_discussion} about the interpretation of bursts). Analytical results have been obtained for a single gene with feedback~\cite{Mackey2011}. Note that this is also precisely the limiting case considered in~\cite{Raj2006}. In absence of interaction between genes, the steady-state distribution of mRNA levels for each gene is a gamma distribution~\cite{Mackey2011,Malrieu2015}, consistently with the derivation made in~\cite{Raj2006}.

In the next section, we focus on exact simulation of model~\eqref{eq_model_bursty} as the simplest compromise for describing the transcriptional bursting property (\cref{fig2}). For mathematical completeness and with a view to proving the algorithm in \cref{sec_proof}, we mention that the model is characterized by its infinitesimal generator:
\begin{equation}\label{eq_generator}
\hspace{-2mm}\begin{array}{rll}
Lf(y,z) & \displaystyle \hspace{-2mm} =  \sum_{i=1}^n \left[ -d_{0,i}y_i \frac{\dr}{\dr y_i} f(y,z) + (s_{1,i}y_i - d_{1,i}z_i) \frac{\dr}{\dr z_i}f(y,z) \right. & \hspace{-3mm}\vspace{2mm}\\
& \qquad\displaystyle + \left. \koni(z) \int_0^{+ \infty} (f(y + he_i, z) - f(y,z)) b_i e^{-b_i h} \intd{h} \right]
\end{array}
\end{equation}
acting on smooth test functions $f:\R_+^n\times\R_+^n\to\R$, where $y=(y_1,\dots,y_n)$ and $z=(z_1,\dots,z_n)$ respectively correspond to values of $M(t)=(M_1(t),\dots,M_n(t))$ and $P(t)=(P_1(t),\dots,P_n(t))$, $e_i\in\R^n$ is the $i$-th vector of the canonical basis, and $b_i = \koffi/s_{0,i}$.
See~\cite{Malrieu2015} for a tutorial introduction to PDMP models, and~\cite{Rudnicki2017} for a more comprehensive reference.
Notably, this form highlights the fact that all genes follow the same individual dynamical model with no direct coupling between bursts, but rather frequency-modulated interactions coming from the dependence of burst frequency $\koni(z)$ on possibly all proteins $z_1,\dots,z_n$.

\section{Simulation algorithm}
\label{sec_algorithm}

In this section, we present a simple algorithm to simulate exact trajectories from the bursty PDMP model~\eqref{eq_model_bursty} with corresponding infinitesimal generator~\eqref{eq_generator}. Here the process is denoted by \( (Y(t), Z(t))_{ t \ge 0} \in \R_+^n \times \R_+^n \) where $Y(t)$ is the vector of mRNA levels and $Z(t)$ is the vector of protein levels.
Formally, the dynamics of this PDMP can be summarized as follows. Given a starting point, the process is driven by a deterministic flow until the first jump time \( T_1 \). Let $\varphi$ denote the flow related to the ODE system \eqref{eq_model_bursty}. Namely, starting from $(y,z) \in \R_+^n \times \R_+^n$, the solution at time $t \ge 0$ is  
\begin{equation}\label{eq_flow_base}
\varphi(y,z,t) = (\flotY(y,z,t),\flotZ(y,z,t)) \in \R_+^n \times \R_+^n
\end{equation}
where for $i \in \{1,\dots,n\}$ :
\[\left\{\begin{aligned}
[\flotY(y,z,t)]_i &= e^{-d_{0,i}t} y_i \\
[\flotZ(y,z,t)]_i &= e^{-d_{1,i} t} z_i + \big(e^{-d_{1,i}t} - e^{-d_{0,i}t}\big) \frac{s_{1,i}}{d_{0,i}-d_{1,i}}y_i
\end{aligned}\right.\]
gives the solution for the mRNA component $y$ and the protein component $z$ (note that we consider $d_{0,i} \neq d_{1,i}$, as derived from the biologically relevant regime $d_{0,i} > d_{1,i}$~\cite{Schwanhausser2011}, but the special case $d_{0,i}=d_{1,i}$ is also analytically solvable).
The distribution of $T_1$ is characterized by
\begin{equation} \label{eq_waiting_time}
\Prob_{y,z}(T_1 > t) = \exp \left( - \int_0^t \sum_{i=1}^n \koni(\flotZ(y,z,\tau)) \intd{\tau} \right)
\end{equation}
where $ \Prob_{y,z} $ denotes the probability conditionally to $ \{ Y(0) = y, Z(0) = z \} $. A new (random) position for the process is drawn, according to the current one. The motion then restarts from this new point, exactly as before. 

Presented this way, the dynamics seems straightforward. However, trying to directly simulate the next waiting time according to \eqref{eq_waiting_time} requires numerical integration, as the corresponding distribution generally does not have a closed form: the algorithm presented here is a natural solution to bypass such integration step. The main underlying assumption is that $\koni$ is continuous and bounded for every $i \in \segn$, which is biologically relevant (note that the bounding assumption is consistent, since the bursty regime already assumes $\koni \ll \koffi$). We can therefore consider a fixed parameter $\lambda$ such that
\begin{equation}\label{eq_jump_rate}
\lambda \ge \sup_{z\in \R_+^n} \Big\{\sum_{i=1}^n \koni(z)\Big\}
\end{equation}
and that will serve as a global \enquote{exponential clock} for computing jump times. From a computational perspective, the optimal value for $\lambda$ is the lowest possible one, which might not be available analytically. In practice, one can simply use $\lambda = \lambda_1 + \cdots + \lambda_n$ where $\koni(z) \le \lambda_i$ are individual bounds. The interest of the simulation method also lies in the fact that the ODE system in~\eqref{eq_model_bursty} can be solved analytically. 

Let $t \ge 0$ be a fixed time. In order to obtain a realisation of a (random) cell state $(Y(t), Z(t))$, we first need to simulate every random event that occurred before time $t$. The bounding assumption allows us to simulate event occurrences according to a homogeneous Poisson process with intensity $\lambda$ and reject some jumps, as in the classical \enquote{thinning} method for simulating non-homogeneous Poisson processes. Given an initial state, a simulation step consists in drawing a waiting time $U$, applying the deterministic flow during a time $U$ and then deciding which type of jump occurs according to the current position of the process: a \enquote{phantom} or a \enquote{true} jump. If it is a phantom jump, the process stays at the same position. On the other hand, if it is a true jump, a new random position is chosen. Namely, we select which gene is impacted by a jump, according to the probability distribution $\mathcal{P}(z)$ on $i\in\{0, 1, \dots, n\}$ defined by
\begin{equation}\label{eq_jump_probability}
p_i(z) = \begin{cases}
\displaystyle 1 - \frac{1}{\lambda} \sum_{i=1}^n \koni(z) & \text{if $i = 0$} \vspace{2mm} \\
\displaystyle \frac{\koni(z)}{\lambda} & \text{if $1 \le i \le n$} \vspace{1mm}
\end{cases}
\end{equation}
and add an exponential random variable of parameter \( b_i \) to the \( i \)-th component of the mRNA vector $Y(t)$. The simulation step is iterated with this new starting point. When the simulation time is beyond $t$, $(Y(t), Z(t))$ follows only the flow applied from the last jump until $t$.

This is formalized in \cref{algo1} and a trajectory of the mRNA component $ Y $ is represented in \cref{fig3}. We give here the algorithm to simulate one cell state $(Y(t), Z(t))$ but it can easily be modified to return cell trajectories collected along a time sequence $(t_1,\dots,t_k)$.

\begin{algorithm}[ht]
\caption{Simulate mRNA and protein levels from the network model~\eqref{eq_model_bursty}}
\label{algo1}
\begin{algorithmic}[1]
\Require initial state $(Y_0,Z_0)$ and final time $t > 0$
\State $Y, Z \gets Y_0,Z_0$ \Comment{\emph{Initialize current state}}
\State $T \gets 0$ \Comment{\emph{Initialize current jump time}}
\While{$T < t$}
    \State $\Yold,\Zold \gets Y,Z$
    \State $\Told \gets T$
    \State $U \gets \Exp(\lambda)$ \Comment{\emph{Draw waiting time using~\eqref{eq_jump_rate}}}
    \State $Y, Z \gets \varphi(\Yold, \Zold, U)$ \Comment{\emph{Apply the deterministic flow~\eqref{eq_flow_base}}}
	\State $i \gets \mathcal{P}(Z)$ \Comment{\emph{Draw $i$ with probability given by~\eqref{eq_jump_probability}}}
\If{$i \neq 0$}
    \State $Y[i] \gets Y[i] + \Exp(b_i)$ \Comment{\emph{Apply jump to mRNA level of gene $i$}}
\EndIf
\State $T \gets T + U$ \Comment{\emph{Update current jump time}}
\EndWhile
\State \Return $\varphi(\Yold, \Zold, t-\Told)$ \Comment{\emph{Extend to final time}}
\end{algorithmic}
\AltTextCMSB{Basic algorithm for exact simulation of the bursty PDMP model.}
\end{algorithm}

This algorithm is computationally efficient since for the same final time required by the user, it requires the generation of much fewer random variables than the SSA model. Although a proper benchmark is out of scope, we mention that for a particular instance of realistic parameter values, the PDMP simulation was found to be approximately 200 times faster than SSA~\cite{Herbach2017}.
Besides, it is also efficient among other PDMP simulation algorithms, as it does not use the distribution \eqref{eq_waiting_time} to simulate inter-arrival times: this would require numerically integrating the jump rate along trajectories, invert the cumulative distribution function and evaluate the inverse at a uniform random variable~\cite{Zeiser2008}.

In some circumstances, computational cost can still be improved. For example, we used the bound $\lambda$ over $\sum_i \koni$, but a bound over $\sum_i \koni \circ \flotZ $ may be sharper, making it possible to reduce the number of phantom jumps. Note that the flow~\eqref{eq_flow_base} has a unique stable attractor: combined with the fact that the jump rate is continuous, one can adapt the algorithm to unbounded jump rates. This requires computing an optimized bound at each iteration, as made for example in the implementation of the Harissa package~\cite{Herbach2023}.
Furthermore, we made a distinction here between $Y(t)$, which is impacted by jumps, and $Z(t)$, which is involved in interactions, but \cref{algo1} can be directly adapted to cases with more variables and other role combinations.

\begin{figure}[ht]
\includegraphics[width=\textwidth]{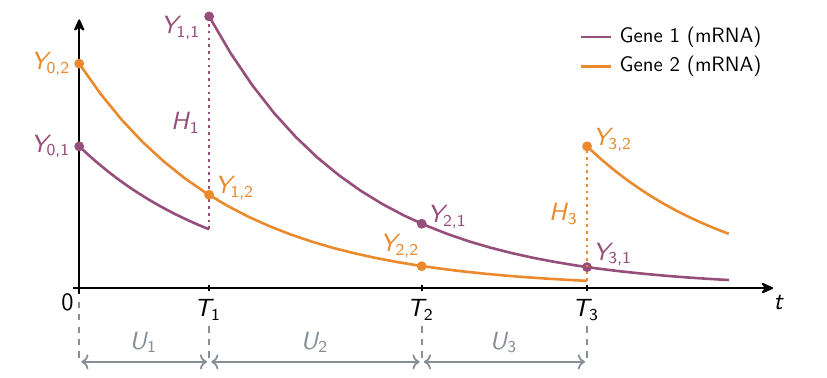}
\caption{Illustration of \cref{algo1} in the case $n=2$, showing the iterative construction for three jump times: here the first jump applies to gene $1$, the second jump is rejected, and the third jump applies to gene $2$. For simplicity, only the mRNA component is shown; the protein component is continuous and derived directly from the corresponding deterministic flow (i.e., using the analytical solution of third equation in~\eqref{eq_model_bursty}).}
\label{fig3}
\AltTextCMSB{Illustration of the simulation algorithm in the case of two genes, showing the iterative construction for three jump times.}
\end{figure}

\section{Discussion}
\label{sec_discussion}

Contrary to the usual SSA, which becomes costly when simulating single-cell gene expression data~\cite{Cannoodt2021}, the PDMP model presented here is much faster to simulate (see, e.g., a factor of 200 in~\cite{Herbach2017}). Mathematically, it can be justified as a rigorous hybrid limit of the discrete formalism when mRNA and protein levels span an important range~\cite{Crudu2009a,Crudu2012,Chen2019}, consistently with biological observations (see \cref{fig2} and~\cite{Schwanhausser2011,Singer2014,Albayrak2016}).
We also mention another interesting link: using the formalism of Poisson representation~\cite{Gardiner1977}, it is possible to show that in absence of interactions, the discrete and hybrid models are in fact related by an exact intertwining relation, with no need for considering a limit regime. Namely, the distribution from the SSA model can then be written as a mixture of Poisson distributions, whose time-dependent mixing distribution is that of the PDMP counterpart~\cite{Jahnke2007,Herbach2019}. This type of result is of major interest when distinguishing measurement and expression models for single-cell data analysis~\cite{Sarkar2021}.

Considering the two-state PDMP model, it is possible to adapt the algorithm presented here to this case, without changing the computational cost. However, we emphasize that only the bursty regime is observed in practice. Note that there is also a basic, \enquote{Poissonian} mode of expression~\cite{Zenklusen2008} but it is described by the \enquote{1-state} model (not 2) with the PDMP being trivial in this case (linear input/output ODE). Basically, the biological phenomenon of transcriptional bursting can be described correctly with the bursty model, not the 2-state model, except when precisely in the bursty regime.

Switching from 2-state PDMP to the bursty PDMP has two advantages: (1) we remove the description of the promoter, which frees up computational memory space, especially when simulating complete trajectories for many cells; (2) we force the correct biological regime, preventing users from accidentally simulating the unrealistic \enquote{saturated} regime ($\koffi < d_{0,i}$, inducing a peak of mRNA at $s_{0,i}/d_{0,i}$ with incorrect tail). In particular, we aggregate the parameters $\koffi$ and $s_{0,i}$, which are not identifiable in any case in the bursty regime.

Considering the simulation of PDMP models, previous authors~\cite{Zeiser2008} proposed a natural algorithm where the waiting time before the next jump is obtained by numerical integration to calculate the survival function of the waiting time, followed by inversion of this function (combined with uniform random variable). We prefer the acceptance-rejection approach because: (1) we are guaranteed exact simulations (unlike numerical integration, no time discretisation is required), and (2) it is easy to code. It is also potentially much faster, since the biologically interesting cases are precisely those where numerical integration is difficult (i.e. requires fine discretisation) because the $k_\text{on}$ function is \enquote{steep} at certain points (i.e., large Lipschitz constant).

As an illustration, we show in \cref{fig4} an example of basic toggle switch network, which generates a bistable pattern emerging from interactions.
We believe this example is interesting because a new \enquote{two-state} layer appears, which seems to be consistent with previous diffusion modeling~\cite{Huynh-Thu2015} while the authors were interpreting the two states as promoter states, possibly from a confusion of time scales and the fact that they had no single-cell data. The interpretation here is that bursts are transcriptional, with proteins tending to buffer these bursts, making the pattern robust.

\begin{figure}[ht]
\includegraphics[width=\textwidth]{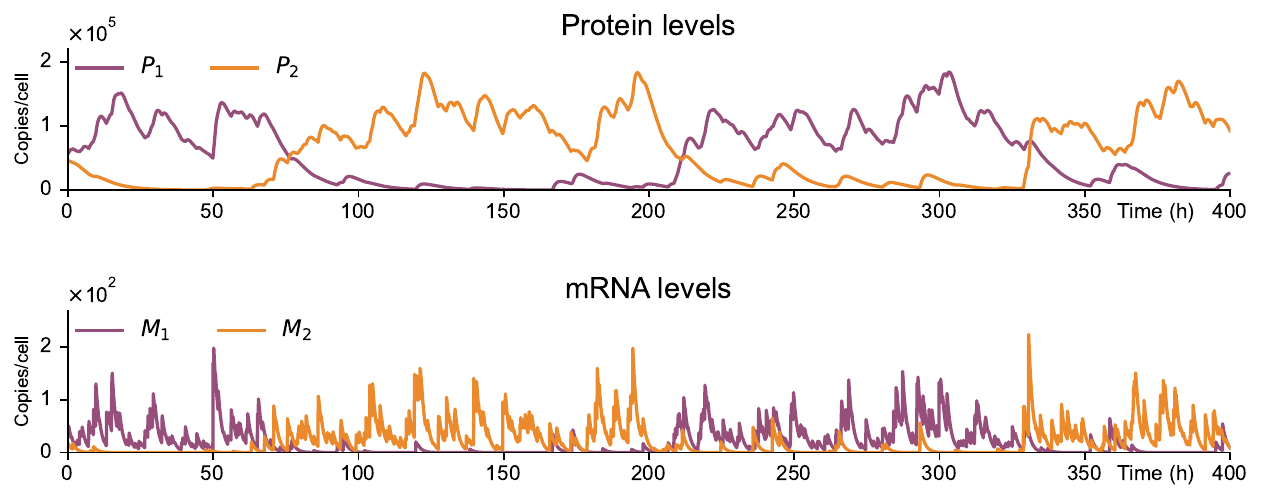}
\caption{Example of trajectory for a toggle switch consisting of two genes that repress each other. This system allows for the emergence of a bistable pattern with distinct \enquote{repressed} ($\koni(P) \approx \qty{0.07}{\per\hour}$) and \enquote{active} ($\koni(P) \approx \qty{1.4}{\per\hour}$) burst frequencies, made robust by proteins having slower degradation rates ($d_{1,i} \approx \qty{0.14}{\per\hour}$) and thus buffering mRNA bursts (with $d_{0,i} \approx \qty{0.7}{\per\hour}$). We argue that, although occurring at a different time scale, the two-state promoter model described in \cref{fig1}A is sometimes confused in the literature with this two-frequency pattern. Importantly, the latter is able to capture bimodal distributions as observed in real data~\cite{Singer2014}, while the two-state promoter per se is only bimodal when $\koffi < d_{0,i}$ (with mRNA \enquote{saturating} at $s_{0,i}/d_{0,i}$) and thus cannot generate bimodality with correct distribution tails~\cite{Herbach2017}.}
\label{fig4}
\AltTextCMSB{Example of trajectory for a toggle switch consisting of two genes that repress each other.}
\end{figure}

Some authors consider a different PDMP describing only proteins~\cite{Friedman2006,Lin2016}: this model can still be rigorously justified as a limit of very stable proteins (i.e., small degradation rate) compared to mRNA, but one must be careful about its interpretation, as it is often stated that in this regime, only 0 or 1 mRNA molecule can be present at a time (so-called \enquote{translational bursts}), which is contradicted by gene expression data showing many mRNA copies at a time.

Another efficient option is to use a diffusion model~\cite{Pratapa2020}, but then there is no biological foundation and in practice such models do not match the bursty model correctly~\cite{Lin2016} and need extra \enquote{ad hoc} noise in order to match real data.

Some more precise data indicate that the \enquote{off} time, i.e. waiting time between bursts, is sometimes non-exponential~\cite{Suter2011}, but \enquote{on} time is well described by exponential: the bursty regime would correspond to a non-exponential waiting time between two bursts, which would be a natural extension of the algorithm presented here.

Finally, thanks to its efficiency, this model can be implemented as an internal component of agent-based frameworks describing a population of spatially interacting cells~\cite{Miles2025}, as already done for example within the Simuscale software~\cite{Nguyen2024}.

\section{Proof}
\label{sec_proof}

In this section we present an explicit construction of the process corresponding to \cref{algo1} and illustrated in \cref{fig3}. We then prove that the infinitesimal generator of the resulting process is that of the PDMP model~\eqref{eq_generator}. This is sufficient to conclude that obtained trajectories are exact. The reader can refer to \cite{Batkai2017} for additional details on the underlying theory. 

In all this section we consider the space of real-valued functions on $\R_+^n \times \R_+^n$ that are continuous and vanish at infinity, equipped with the uniform norm. Namely,
\[ \Vert f \Vert_{\infty} = \sup_{(y,z) \in \R_+^n \times \R_+^n} \vert f(y,z) \vert. \]

\Cref{algo1} leads to the following construction. We first define a discrete-time Markov chain $(Y_k, Z_k)_{k \ge 0}$ where $(Y_k, Z_k) \in \R_+^n \times \R_+^n$, which corresponds to the right limit of the process $(Y(t), Z(t))_{t \ge 0}$ at specific (random) time points. These time points are given by a homogeneous Poisson process $ (N(t))_{t \ge 0} $ with intensity $ \lambda $. Let $(T_k)_{k \ge 0}$ denote its jump time sequence and $(U_k)_{k \ge 1}$ its inter-arrival times. Let $ Y_0 $ and $ Z_0 $ be random variables independent of $ (N(t))_{t \ge 0} $. For $k\ge 0$, we recursively define
\[\left\{\begin{aligned}
Y_k &= \flotY(Y_{k-1}, Z_{k-1}, U_k) + S_k \\
Z_k &= \flotZ(Y_{k-1}, Z_{k-1}, U_k)
\end{aligned}\right.\]
where $S_k$ is a random vector constructed as follows: at each step $k$, we draw a random index $i\in\{0,1,\dots,n\}$ according to the probability distribution \( \mathcal{P}(Z_k) \) defined at \eqref{eq_jump_probability}. When \( i \neq 0 \), we consider an additional variable \(E\) with exponential distribution $\mathcal{E}(b_i)$. We set
\[S_k = \begin{cases}
0 & \text{if $i = 0$} \vspace{2mm} \\
E e_i & \text{if $1 \le i \le n$} \vspace{1mm}
\end{cases}\]
where $ e_i $ is the \( i \)-th vector of the canonical base of $ \R^n $. Finally, we define $ (Y(t), Z(t))_{t \ge 0} $ for every $ t \ge 0 $ by applying the flow between jump times. Namely,
\[ \forall t \in [T_k, T_{k+1}), \quad (Y(t), Z(t)) = \varphi(Y_{k-1}, Z_{k-1}, t - T_k). \]
The following result states that the constructed process \( (Y(t), Z(t))_{t \ge 0} \) corresponds to the bursty model.

\begin{theorem}
The infinitesimal generator of the semigroup \( (T(t))_{t \ge 0} \) induced by \( (Y(t), Z(t))_{t \ge 0} \) is the operator defined at~\eqref{eq_generator}. Moreover the semigroup \( (T(t))_{t \ge 0} \) is strongly continuous.
\end{theorem}

\begin{proof}
Let $\Esp_{y,z}$ denote the conditional expectation given initial data $ \Esp[ \,\cdot\, | Y_0 = y, Z_0 = z] $. The process \( (Y(t), Z(t))_{t \ge 0} \) is characterized by its semigroup \( (T(t))_{t \ge 0} \) where for every \( t \ge 0 \) the operator $ T(t) $ is defined by 
\[ T(t)f(y,z) = \mathbb{E}_{y,z} \left[ f(Y(t), Z(t)) \right]. \]
The infinitesimal generator of \( (T(t))_{t \ge 0} \) is defined by
\[ \lim_{t \rightarrow 0}  \frac{T(t)f - f}{t}, \]
for $ f $ such that the limit exists. To calculate this limit, we first obtain an explicit form for $(T(t))_{t \ge 0}$. In order to do that we first need the transition operator $ P $ of the Markov chain $ (Y_k, Z_k)_{k \ge 0} $. 
Let us define the operators $ J $ and $ K_t $ by 
\[ Jf(y, z) = f(y, z) + \frac{1}{\lambda} \sum_{i=1}^n \koni(z) \int_0^{+ \infty} (f(y + he_i, z) - f(y, z)) b_ie^{-b_ih} \intd{h} \]
and
\[K_t f(y,z) = f(\varphi(y,z,t)) .\] 
By definition, the transition operator $ P $ of $ (Y_k, Z_k)_{k \ge 0} $ is
\[ Pf(y,z) = \Esp_{y,z} \left[ f(Y_1, Z_1) \right]. \]
Using the recursive definition of $ Y_1 $ and $ Z_1 $ we have 
\begin{align*} 
\mathbb{E}_{y,z} \left[ f(Y_1, Z_1) \vert U_1 = u \right] &= \mathbb{E}_{y,z} \left[ f( \flotY(y, z, u) + S_1, \flotZ(y,z, u)) \right] \\
&= K_uJf(y, z).
\end{align*}
Thus the operator $ P $ is given by 
\[ Pf(y,z) = \Esp_{y ,z} \left[ K_{U_1} Jf(y,z) \right] = \int_0^{+ \infty} K_s Jf(y,z) \lambda e^{- \lambda s} \intd{s}. \]
Thanks to the previous construction, we have
\begin{equation} \label{eq_semigroup}
T(t)f = \sum_{n \ge 0} \mathbb{E} \left[\mathds{1}_{\{N(t) = n\}} K_{U_1}J \circ \cdots \circ K_{U_n}J \circ K_{t - T_n}f \right].
\end{equation}
Indeed, conditionally to having $ n $ events of the underlying Poisson process \( (N(t))_{t \ge 0} \), $ f(Y(t), Z(t)) $ is only the flow operator applied from the last jump until $ t $. Namely 
\[ \Esp_{y,z} \left[ f(Y(t), Z(t)) \vert N(t) = n \right] = \Esp_{y,z} \left[ K_{t - T_n} f(Y_n, Z_n) \right]. \]
Furthermore, for every $ k \ge 1 $, 
\begin{align*} 
\Esp_{y,z} \left[ f(Y_k, Z_k) \right] &= \Esp_{y,z} \left[ Pf(Y_{k-1}, Z_{k-1}) \right] \\
&= \Esp_{y,z} \left[ K_{U_{k-1}} J f(Y_{k-1}, Z_{k-1}) \right].
\end{align*}
Integrating this relation, backwards from the $n$-th jump to 0, leads to the formula at \eqref{eq_semigroup}. Let us now calculate the infinitesimal generator of \( (T(t))_{t \ge 0} \). Note that jumps (including phantom ones) arrived according to the Poisson process $ (N(t))_{t \ge 0} $. Thus the probability of having two jumps in the interval \( [0,t] \) decreases faster than $ t $ as $ t \rightarrow 0 $. This is why we isolate the first two terms in \eqref{eq_semigroup} and get  
\begin{align}
T(t)f &= e^{- \lambda t} K_tf + \mathbb{E} \left[ \mathds{1}_{U_1 \le t} \mathds{1}_{U_2 \ge t - U_1} K_{U_1}J K_{t - U_1}f \right] + R(t,f) \nonumber \\
&= e^{- \lambda t} K_tf + \lambda e^{- \lambda t} \int_0^t K_uJ K_{t - u} \intd{u} + R(t,f). \label{eq_dvp_expression_semigroup}
\end{align}
As stated,
\[ \Vert R(t,f) \Vert_{\infty} \le \Vert f \Vert_{\infty} \Prob(N(t) > 1) = \Vert f \Vert_{\infty} (1 - e^{-\lambda t}(1 + \lambda t)) \]
so 
\[ \frac{\Vert R(t,f)\Vert_{\infty}}{t} \rightarrow 0 \quad \text{when} \quad t \rightarrow 0.\]
We then use Taylor expansion on the term $ e^{-\lambda t} $ and simplify to get
\[ \frac{T(t)f - f}{t} = \frac{K_tf - f}{t} - \lambda K_tf  + \frac{\lambda}{t} \int_0^t K_uJK_{t-u}f\intd{u} + o(1). \]
To conclude the proof, we need to show that
\[ \frac{K_tf - f}{t} \to \sum_{i=1}^n \langle F_i, \nabla_if \rangle, \quad K_tf \to f, \quad \text{and} \quad \frac{1}{t} \int_0^t K_uJK_{t-u}f\intd{u} \to Jf \]
uniformly when \( t \rightarrow 0 \). We are exactly in the setting of~\citep[Chapter 16, assumptions page 253]{Batkai2017}. The convergence of the first term (usually called Koopman semigroup) is done at~\cite[pages 255-256]{Batkai2017}. The two other convergences can be proved using the same two ingredients: the uniform continuity of $ f $ and the Lipschitz continuity of the flow in its first variable. The uniform continuity of $ f $ is immediate given the chosen space of test functions, and the property on the flow is a consequence of the vector field regularity.
Finally, after simplification, we obtain 
\[ \frac{T(t)f - f}{t} \to Lf, \]
where $ L $ is the operator defined at \eqref{eq_generator}. In order to show that \( (T(t))_{t \ge 0} \) is strongly continuous, it is sufficient to show that for every \( t \ge 0 \) and for every function \( f \), \( T(t)f \) is continuous and that \( \Vert T(t)f - f \Vert \) tends to 0 as \( t \rightarrow 0 \). The regularity of the flow \( \varphi \) to the initial condition gives the continuity of \( Kf \). Combined with the continuity of \( Jf \), we obtain the continuity of \( T(t)f \). The second part is the result of \eqref{eq_dvp_expression_semigroup} and the fact that \( K_tf \rightarrow f \) uniformly when \( t \rightarrow 0\). Note that the strong continuity of the semigroup \( (T(t))_{t \ge 0} \) implies that it is characterized by its infinitesimal generator \( L \) (see \citep[page 122]{Batkai2017}).
\end{proof}

\phantomsection
\addcontentsline{toc}{section}{Code Availability}
\paragraph{Code Availability.}\hspace{-2mm}
An implementation of the simulation algorithm in Python is available within the Harissa package at \url{https://github.com/ulysseherbach/harissa} along with tutorial notes.

\phantomsection
\addcontentsline{toc}{section}{Acknowledgements}
\paragraph{Acknowledgements.}\hspace{-2mm}
The second author is very grateful to Asmaa Labtaina for preliminary work on this subject.
We also thank the anonymous reviewers for their constructive comments.

\phantomsection
\addcontentsline{toc}{section}{Funding}
\paragraph{Funding.}\hspace{-2mm}
This work was supported in part by the Programmes et Équipements Prioritaires de Recherche (PEPR) Santé Numérique under Project 22-PESN-0002.

\phantomsection

\end{document}